\documentclass[pra,showpacs,twocolumn]{revtex4-1}

\usepackage{amsmath}
\usepackage{amsfonts}
\usepackage{amssymb}
\usepackage{mathrsfs}
\usepackage{latexsym}
\usepackage{dsfont}
\usepackage{graphicx}

\newenvironment{definition}[1][Definition]{\vspace{0.3cm}\noindent\textbf{#1.} }{\vspace{0.2cm}}
\newtheorem{theorem}{Theorem}
\newenvironment{proof}[1][Proof]{\noindent\textbf{#1.} }{\ \rule{0.5em}{0.5em}}

\begin{document}

\title{Witnessing nonclassical multipartite states}

\author{A. Saguia}
\email{amen@if.uff.br}
\affiliation{Instituto de F\'{\i}sica, Universidade Federal Fluminense,
Av. Gal. Milton Tavares de Souza s/n, Gragoat\'a, 24210-346, Niter\'oi, RJ, Brazil.}

\author{C. C. Rulli}
\affiliation{Instituto de F\'{\i}sica, Universidade Federal Fluminense,
Av. Gal. Milton Tavares de Souza s/n, Gragoat\'a, 24210-346, Niter\'oi, RJ, Brazil.}

\author{Thiago R. de Oliveira}
\affiliation{Instituto de F\'{\i}sica, Universidade Federal Fluminense,
Av. Gal. Milton Tavares de Souza s/n, Gragoat\'a, 24210-346, Niter\'oi, RJ, Brazil.}

\author{M. S. \surname{Sarandy}}
\affiliation{Instituto de F\'{\i}sica, Universidade Federal Fluminense,
Av. Gal. Milton Tavares de Souza s/n, Gragoat\'a, 24210-346, Niter\'oi, RJ, Brazil.}

\date{\today }

\begin{abstract}
We investigate a witness for nonclassical multipartite states based on their 
disturbance under local measurements. The witness operator provides a 
sufficient condition for nonclassicality that coincides with a nonvanishing 
global quantum discord, but it does not demand an extremization procedure. 
Moreover, for the case of $Z_2$-symmetric systems, we rewrite the witness in 
terms of correlation functions so that classicality is found to necessarily require either vanishing 
magnetization in the invariant axis  or isotropy of the two-point function in the transverse spin plane. 
We illustrate our results in quantum spin chains, where a characterization of factorized 
ground states (with spontaneously broken $Z_2$ symmetry) is achieved. As a by-product, 
the witness will also be shown to indicate a second-order quantum phase transitions, which 
will be illustrated both for the $XY$ and Ashkin-Teller spin chains. 
\end{abstract}


\maketitle

\section{Introduction}

The characterization of quantum correlated multipartite systems has been 
a central topic in quantum information theory~(see, e.g., Refs.~\cite{Piani:08,%
Acin:10,Bennett:11} for recent discussions). Such systems are described by quantum 
states that cannot be represented by joint probability distributions. Remarkably, 
the existence of quantum correlations may occur even in the absence of entanglement  
and still be a resource for a number of quantum information applications (see, e.g., 
Refs.~\cite{Zurek:03,locking,DQC1,Brodutch:11,Fanchini:11}). Measures of 
quantum correlations provide then a valuable tool to recognize useful quantum states 
for a certain task. In this context, an important  measure is given by quantum discord 
(QD)~\cite{Ollivier:01}. It arises as a difference between two expressions for the 
total correlation in a bipartite system (as measured by the mutual information), 
which are classically equivalent but distinct in the quantum regime. QD has been 
generalized to the multipartite scenario by different approaches~\cite{Modi:10,%
Rulli:11,Chakrabarty:10,Okrasa:11}. Unfortunately, both bipartite and multipartite 
versions of QD involve extremization procedures, which makes their computation 
viable only for very small systems or for particular classes of states. 

In this context, a possible strategy to identify composite systems that do not exhibit 
purely classical correlations is to look for a witness operator, namely, an observable 
able to detect the presence of nonclassical states. A witness $\cal{W}$ is  defined here 
as a Hermitian operator whose norm $\| {\cal W} \|$ is such that, for all nonclassical states, 
$\| {\cal W} \| > 0$, with a convenient norm measure adopted. Therefore, $\| {\cal W} \| > 0$ 
is a sufficient condition for nonclassicality (or, equivalently, a necessary condition for classicality). 
Several investigations for establishing conditions to witness bipartite nonclassical 
states have recently been carried out~\cite{Bylicka:10,Ferraro:10,Dakic:10,Maziero:10}. 
Our aim here is to provide a systematic approach to obtaining a witness for nonclassicality in 
multipartite states, which will be based on their disturbance under local measurements. 
In turn, such a witness will naturally detect states with a nonvanishing multipartite QD. As an illustration, 
we will consider arbitrary $Z_2$-symmetric states, where the witness operator can be conveniently rewritten in 
terms of correlation functions so that classicality is found to 
require either vanishing magnetization in the $Z$ direction or isotropy of the two-point function in the XY spin plane. 
In particular, we will apply such results in quantum spin chains, where it will be shown that $\| {\cal W} \|$ exhibits a completely 
different behavior in the cases of $Z_2$-symmetric states and states with spontaneous symmetry breaking (SSB). In the latter 
case, a characterization of nontrivial factorized ground states (with $Z_2$ SSB) will be achieved. As a by-product, 
the witness will also be shown to exhibit a nonanalytical behavior at second-order quantum phase transitions. This nonanalyticity 
will be revealed by the first derivative of $\| {\cal W} \|$ and will be used as a tool to characterize second-order quantum 
critical points in both the XY and Ashkin-Teller spin chains. 

\section{Witness for nonclassical states}

Let us begin by considering a composite system in an $N$-partite Hilbert space 
${\cal H} = {\cal H}_{A_1} \otimes {\cal H}_{A_2} \otimes \cdots \otimes {\cal H}_{A_N}$. 
The system is characterized by quantum states described by density operators $\rho \in {\cal B}({\cal H})$, where 
${\cal B}({\cal H})$ is the set of bound, positive-semidefinite operators acting on 
${\cal H}$ with trace given by ${\textrm{Tr}}\,\rho=1$. In order to define a multipartite 
classical state, we will generalize the idea of nondisturbance under projective measurements 
introduced for bipartite states in Ref.~\cite{Luo:08}. Indeed, we will denote a set of local 
von Neumann measurements as $\{\Pi_j\} = \{\Pi_{A_1}^{i_1} \otimes \cdots \Pi_{A_N}^{i_N}\}$, 
with  $j$ denoting the index string $(i_1 \cdots i_N)$. Then, after a non-selective measurement, 
the density operator $\rho$ becomes
\begin{equation}
\Phi(\rho) = \sum_j \Pi_j \rho \Pi_j \, .
\label{vn}
\end{equation}
This operation can then be used to define a classical state. 

\begin{definition}
If there exists any measurement $\{\Pi_j\}$ such that $\Phi(\rho)=\rho$ then $\rho$ describes a {\it classical} state 
under von Neumann local measurements.
\end{definition}

Therefore, it is always possible to find out a local measurement basis such that a classical state 
$\rho$ is kept undisturbed. In this case, we will denote $\rho \in {\cal C}^N$, where ${\cal C}^N$ 
is the set of $N$-partite classical states. Observe that such a definition of classicality coincides 
with the vanishing of global QD~\cite{Rulli:11}. Note also that $\Phi(\rho)$ by itself 
is a classical state for any $\rho$, namely, $\Phi(\Phi(\rho)) = \Phi(\rho)$. Hence, $\Phi(\rho)$ can 
be interpreted as a decohered version of $\rho$ induced by measurement. A witness for nonclassical states 

can be directly obtained from the observation that the elements of the set $\{\Pi_j\}$ are eigenprojectors of $\rho$. 
This can be shown by the theorem below (see also Ref.~\cite{Luo:08}).

\begin{theorem}
\,\,$\rho \in {\cal C}^N \Longleftrightarrow \left[ \rho, \Pi_j \right] = 0 \,\,\,(\forall j)$, with 
$\Pi_j = \Pi_{A_1}^{i_1} \otimes \cdots \Pi_{A_N}^{i_N}$ and $j$ denoting the index string 
$(i_1 \cdots i_N)$.
\label{t1}
\end{theorem}
\begin{proof}
If $\rho \in {\cal C}^N$ then $\Phi(\rho)=\rho$. Then, similarly as in Ref.~\cite{Luo:08}, a direct 
evaluation of $\sum_j \left[\rho,\Pi_j\right]\,\left[\rho,\Pi_j\right]^{\dagger}$ yields
$\sum_j \left[\rho,\Pi_j\right]\,\left[\rho,\Pi_j\right]^{\dagger} = \Phi(\rho^2)-\rho^2$. 
However, $\Phi(\rho)=\rho$ also implies that $\Phi(\rho^2)=\rho^2$. Therefore, 
\begin{equation}
\sum_j \left[\rho,\Pi_j\right]\,\left[\rho,\Pi_j\right]^{\dagger} = 0\,.
\end{equation}
Hence, $\rho \in {\cal C}^N \Longrightarrow \left[\rho,\Pi_j\right]=0$. 
On the other hand, if $\left[\rho,\Pi_j\right]=0$ then $\{\Pi_j\}$ provides a 
basis of eigenprojectors of $\rho$. Then, from the spectral decomposition, we obtain
\begin{equation} 
\rho = \sum_j p_j \Pi_j = \sum_{i_1,\cdots,i_N} p_{i_1,\cdots,i_N} \Pi_{A_1}^{i_1} \otimes \cdots \Pi_{A_N}^{i_N}\, ,
\label{rho-sd}
\end{equation}
which immediately  implies that $\Phi(\rho)=\rho$. Hence $\left[\rho,\Pi_j\right]=0 \Longrightarrow \rho \in {\cal C}^N$.
\end{proof}

We can now propose a necessary condition to be obeyed for arbitrary multipartite classical states.

\begin{theorem}
Let $\rho$ be a classical state and $\rho_{A_i}$ the reduced density operator for the subsystem $A_i$. 
Then $\left[\rho,\rho_{A_1} \otimes \cdots \otimes \rho_{A_N}\right] = 0$.
\end{theorem}
\begin{proof}
From theorem~\ref{t1}, if $\rho \in {\cal C}^N$ then the spectral decomposition of $\rho$ yields Eq.~(\ref{rho-sd}). 
Therefore, 
\begin{equation}
\rho_{A_i} = \sum_{n} p_n \Pi_{A_i}^{n} \,.
\end{equation}
Hence, by direct evaluation, we obtain that $\left[\rho,\rho_{A_1} \otimes \cdots \otimes \rho_{A_N}\right] = 0$.
\end{proof}

Observe that, given a composite multipartite state $\rho$, it is rather simple to evaluate the commutator 
$\left[\rho,\rho_{A_1} \otimes \cdots \otimes \rho_{A_N}\right]$, with no extremization procedure as usually 
required by QD computation. From the necessary condition for classical states above, we can define 
a witness for nonclassicality by the norm of the operator $\left[\rho,\rho_{A_1} \cdots \rho_{A_N}\right]$. 
Indeed, if $\rho \in {\cal C}^N \Longrightarrow \| {\cal W} \| = 0$, where 
\begin{equation}
{\cal W} = \left[\rho,\rho_{A_1} \otimes \cdots \otimes \rho_{A_N}\right] \, .
\end{equation}
Therefore, $\| {\cal W} \| > 0$ is sufficient for nonclassicality. 

For concreteness, we will take $\| {\cal W} \|$  as defined by the trace norm, namely, 
\begin{equation}
\| {\cal W} \| = {\textrm{Tr}} \sqrt{ {\cal W} {\cal W}^\dagger} \, .
\end{equation}
The witness $\| {\cal W} \|$ provides a sufficient -- but not necessary -- condition for nonclassicality. 
As a simple example, consider a two-qubit system $AB$ prepared in 
a Bell state $|\psi\rangle = (|00\rangle + |11\rangle)/\sqrt{2}$, where $\{|0\rangle,|1\rangle\}$ denotes the computational 
basis for both qubits $A$ and $B$. The reduced density operators for each qubit are given by maximally mixed states, 
with $\rho_A \otimes \rho_B$ proportional to the identity operator. Therefore, by taking $\rho = |\psi\rangle\langle\psi|$, we have that 
$\left[\rho,\rho_{A} \otimes \rho_{B}\right] = 0$, even though $\rho$ describes a maximally entangled state. However, as we will show below, 
such a witness can be a very useful tool to investigate nonclassicality in other (non-maximally entangled) classes of quantum correlated states 
in many-body systems.

\section{Nonclassical states in  bipartite systems with  $Z_2$ symmetry} 
As an application of the nonclassicality witness, let us translate the witness operator 
$\| {\cal W} \|$ in terms of correlation functions and magnetization in quantum spin systems.  
To this aim, we will consider an interacting pair of spins-1/2 particles in a chain, 
which is governed by an arbitrary Hamiltonian $H$ that is both real and exhibits $Z_2$ symmetry, 
i.e. invariance under $\pi$-rotation around a given spin axis. 
By taking this spin axis as the $z$ direction, this implies the commutation of $H$ with the  parity 
operator $\bigotimes_{i=1}^{N} \sigma^3_i$, where $N$ denotes the total number of spins and 
$\sigma^3_i$ is the Pauli operator along the $z$-axis at site $i$. From this symmetry, the two-spin reduced 
density matrix at sites labelled by $i$ and $j$ in the basis 
$\{ |\uparrow\uparrow\rangle, |\uparrow\downarrow\rangle, |\downarrow\uparrow\rangle, |\downarrow\downarrow\rangle\}$, 
with $|\uparrow\rangle$ and $|\downarrow\rangle$ denoting the eigenstates of  $\sigma^3$, will be given by an $X$-state, 
reading
\begin{equation}
\mathcal{\rho}_{AB}=\left( 
\begin{array}{cccc}
a & 0 & 0 & f \\ 
0 & b_1 & z & 0 \\ 
0 & z & b_2 & 0 \\ 
f & 0 & 0 & d
\end{array}
\right) .  \label{rhoAB}
\end{equation}
In terms of spin correlation functions, these elements can be written as 
\begin{eqnarray}
a &=& \frac{1}{4} \left(1+G^i_{z}+G^{j}_z+G^{ij}_{zz}\right) \, , \hspace{0.3cm} 
z = \frac{1}{4} \left(G^{ij}_{xx}+G^{ij}_{yy} \right) \, , \nonumber \\
b_1 &=& \frac{1}{4} \left(1+G^i_{z}-G^{j}_z-G^{ij}_{zz}\right) \, , \hspace{0.3cm}
f = \frac{1}{4} \left(G^{ij}_{xx}-G^{ij}_{yy} \right) \nonumber \, , \\
b_2 &=& \frac{1}{4} \left(1-G^i_{z}+G^{j}_z-G^{ij}_{zz}\right) \, , \nonumber \\
d &=& \frac{1}{4} \left(1-G^i_{z}-G^{j}_z+G^{ij}_{zz}\right) \, , 
\label{relem}
\end{eqnarray}
where $G^k_{z} = \langle \sigma_z^k \rangle$ $(k=i,j)$ is the magnetization density at site $k$ and 
$G^{ij}_{\alpha\beta}=\langle \sigma^i_\alpha \sigma^j_\beta \rangle$ 
($\alpha,\beta=x,y,z$) denote two-point spin-spin functions at sites $i$ and $j$, with the expectation 
value taken over the quantum state of the system. Note that, in case of translation invariance, 
which will be assumed here for simplicity, $G^k_{z}=G^{k^\prime}_{z} \equiv G_{z} $ ($\forall\, k,k^\prime$) and, 
therefore, $b_1=b_2 \equiv b$.  In order to evaluate the witness in terms of the magnetization density 
and spin-spin correlation functions, we compute the commutator ${\cal W} = \left[\rho,\rho_{A} \otimes \rho_{B}\right]$, 
yielding
\begin{equation}
{\cal W} = \left( 
\begin{array}{cccc}
0 & 0 & 0 & k \\ 
0 & 0 & 0 & 0 \\ 
0 & 0 & 0 & 0 \\ 
-k & 0 & 0 & 0
\end{array}
\right) ,  \label{commut-X}
\end{equation}
with $k = f \left[ \left(b+d\right)^2 - \left(a+b)^2\right)\right]$. Therefore
\begin{equation}
\| {\cal W} \| = \frac{1}{2} \, | \langle \sigma^i_x \sigma^j_x \rangle - \langle \sigma^i_y \sigma^j_y \rangle | \,\,
| \langle \sigma_z \rangle |.
\end{equation}
Therefore, for translation invariant systems with $Z_2$ symmetry, the necessary conditions for a state be classical are 
either vanishing magnetization in the $z$ spin direction or isotropy in the $XY$ spin plane, i.e., 
\begin{equation}
\langle \sigma_z \rangle = 0 \,\,\,\, {\textrm{or}} \,\,\,\, 
\langle \sigma^i_x \sigma^j_x \rangle = \langle \sigma^i_y \sigma^j_y \rangle.
\label{witness-cf}
\end{equation}
Hence, the violation of the condition established in Eq.~(\ref{witness-cf}) is sufficient 
for the nonclassicality of the quantum state.

\section{Factorized ground states in the XY model} 

As an example of a $Z_2$-symmetric system, let us consider a XY spin-1/2 chain, which is 
governed by the following Hamiltonian
\begin{equation}
H=-\sum_{i=0}^{N-1}\left\{  \frac{J}{2}\left[  (1+\gamma)\sigma_{i}%
^{x}\sigma_{i+1}^{x}+ (1-\gamma)\sigma_{i}^{y}\sigma_{i+1}^{y}\right]
+h\sigma_{i}^{z}\right\}  , \label{XY}%
\end{equation}
with $N$ being the number of spins in the chain, $\sigma_{i}^{m}$ the $i$-th spin
Pauli operator in the direction $m=x,y,z$ and periodic boundary conditions 
assumed. For the (dimensionless) anisotropy parameter $\gamma\rightarrow0$, the model reduces to the XX model 
whereas for all the interval $0<\gamma\le1$ it belongs to the Ising universality class, 
reducing to the transverse field Ising model at $\gamma=1$.  
The parameter $h$ is associated with the
external transverse magnetic field. In the thermodynamical limit, for $\lambda\equiv J/h = 1$, a
second-order quantum critical line takes place for any $0<\gamma\le1$.

The exact solution of the XY model is well known~\cite{XYsol-1,XYsol-2}. The 
Hamiltonian~(\ref{XY}) can be diagonalized via a Jordan-Wigner map followed by a
Bogoliubov and Fourier transformation. Given the $Z_2$ symmetry, the two-spin reduced density operator can
be written as
\begin{equation}
\rho = \frac{1}{4} \left[ I\otimes I + \sum_{i=1}^{3} \left( c_i \sigma^i \otimes \sigma^i \right) +
c_4 \left(I \otimes \sigma^3 + \sigma^3 \otimes I\right) \right], 
\label{RDO}
\end{equation}
where, in terms of the parametrization given in Eq.~(\ref{rhoAB}), we have 
$c_1 = 2z+2f$, $c_2 = 2z-2f$, $c_3 = a+d-2b$, and $c_4 = a-d$ (with $b_1=b_2\equiv b$).
Due to the fact that the system is translationally invariant, the reduced state~(\ref{RDO}) 
depends only on the distance between the spins. We will consider here nearest-neighbor pairs. 
The magnetization density $G_{z}=\langle \sigma_z \rangle$ and the two-point functions 
$G^{i,i+1}_{\alpha\beta}=\langle \sigma^i_\alpha \sigma^{i+1}_\beta \rangle$ can be directly 
obtained from the exact solution of the model~\cite{XYsol-1,XYsol-2}.

In the ordered ferromagnetic phase ($\lambda>1$), a finite magnetization $\langle \sigma^x_i \rangle$ (in the coupling direction) 
emerges, breaking the $Z_2$ symmetry of the Hamiltonian for an infinite chain. 
This is possible because, at the critical point, the ground state 
becomes two-fold degenerated, exhibiting opposites values of $\langle \sigma^x_i \rangle$  
(denoted as $|+m\rangle$ and $|-m\rangle$). Naturally, a superposition of the two 
ground states is also a ground state, which could then preserve the $Z_2$ symmetry. However,
in the thermodynamic limit, these two degenerated ground states are not connected by local unitaries.
Thus, to go from one to the other, one needs to flip all the spins at the same time, which is both a  
highly non-local and energy costly operation. Therefore,  a small external perturbation will 
always pick up one of the two states, and the system will be "frozen" at it. The use of a superposition  
that preserves the symmetry could be justified in the case of a finite system, where the system, if 
highly isolated, could exist in such state. Another possibility would be to consider an equal mixture of the
two ground states, which could be reasonable in the case one prepares a system at some temperature
$T$ and cools it down to the limit of $T\rightarrow 0$. This is a $Z_2$-symmetric state called 
{\it thermal ground state}. Note that, since the thermodynamic properties depend only on the
energy spectrum (and not on the eigenstates), all choices of ground states will be equivalent. 
However, classical and quantum correlations {\it do} depend on the eigenstates and, therefore, 
the symmetric and broken choices of ground states may exhibit different pattern of 
correlations~\cite{Oliveira:08,Giampaolo:08,Ciliberti:10,Tomasello:10}.

Remarkably, the ground state with SSB exhibits a factorization (product state) point at 
\begin{equation}
\gamma^2 + \left(\frac{1}{\lambda}\right)^2 = 1.
\label{factorization}
\end{equation}
This point identifies a change in the behavior of the correlation functions decay, which pass from
monotonically to an oscillatory decay~\cite{XYsol-2}. Besides, it has been realized
that, at this point, the product ground state is two-fold degenerated even for finite systems~\cite{Giampaolo:08,Ciliberti:10,Tomasello:10}.

In order to apply the witness $\| {\cal W} \|$ in the XY model, we will first consider the case of the thermal ($Z_2$-symmetric) 
ground state, for which the reduced density matrix is given by Eq.~(\ref{RDO}). For this case, a plot of $\| {\cal W} \|$ as a 
function of $\lambda$ is provided in Fig.~\ref{f1} for $\gamma=0.6$. Observe that this plot shows that the only 
possible classical state appears for $\lambda=0$. Indeed, this point corresponds to an infinite transverse magnetic field applied, 
which leads the system to a product state, with all spins individually pointing in the $z$ direction. Concerning the 
factorization point given in Eq.~(\ref{factorization}), it is absent for the thermal state, since the statistical mixture of 
{\it non-orthogonal} product states exhibits nonvanishing quantum correlations~\cite{Ciliberti:10}.
\begin{figure}[ht]
\centering {\includegraphics[clip,scale=0.3]{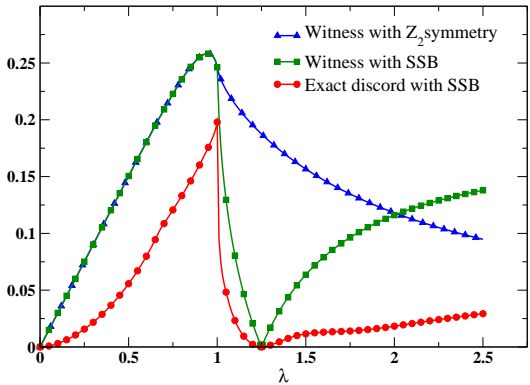}}
\caption{(Color online) Witness $\| {\cal W} \|$ and exact value of QD for a pair of spins 
in both symmetric and broken ground states of the XY spin-1/2 infinite chain as a function of $\lambda$ for $\gamma=0.6$.}
\label{f1}
\end{figure}

On the other hand, if SSB is taken into account, which must be the case in the thermodynamic limit, this 
picture dramatically changes. Indeed, in the broken case, the two-spin reduced density matrix is more
complicated and reads
\begin{equation}
\mathcal{\rho}_{AB}=\left( 
\begin{array}{cccc}
a & p & p & f \\ 
p & b & z & q \\ 
p & z & b & q \\ 
f & q & q & d
\end{array}
\right) ,  \label{rhoAB-SSB}
\end{equation}
where the functions $p(\lambda,\gamma)$ and $q(\lambda,\gamma)$ are given by
\begin{eqnarray}
p(\lambda,\gamma) &=& \langle \sigma^i_x \rangle + \langle \sigma^i_x \sigma^{i+1}_z \rangle \, ,\\
q(\lambda,\gamma) &=& \langle \sigma^i_x \rangle - \langle \sigma^i_x \sigma^{i+1}_z \rangle \, .
\end{eqnarray}

The evaluation of these functions have been detailed discussed in Ref.~\cite{Oliveira:08}, where  
tight lower and upper bounds for $p(\lambda,\gamma)$ and $q(\lambda,\gamma)$ are obtained. In our plots, 
either of such bounds essentially yields the same curves, which led us to keep the results produced with 
the lower bound. Concerning the results for $\| {\cal W} \|$, we can compare the plots for the thermal ground  
state and the ground state with SSB also in Fig.~\ref{f1}. Observe that, in the case with SSB, the only point for 
which $\| {\cal W} \|=0$ is $\lambda = 1.25$. Indeed, this corresponds to a classical state, as can 
be confirmed by the exact computation of the symmetric QD. Remarkably, this classical state is associated 
with one of two doubly-degenerated factorized (fully product) ground state, which appears as a consequence of 
the $Z_2$ SSB. This characterization of factorized ground states through $\| {\cal W} \|=0$ can be numerically confirmed in 
Fig.~\ref{f2}, where $\| {\cal W} \|$ is plotted for a number of different values of the parameter $\gamma$. 
These unique classical points for the XY model are then clearly revealed by the witness evaluation. 
Another aspect of the witness observed from Figs.~\ref{f1} and \ref{f2} is that it exhibits nonanalyticity in its first derivative 
at the second-order quantum critical point $\lambda=1$. This phenomenon occurs both for the symmetric and broken ground states, 
which promotes $\| {\cal W} \|$  to a simple useful tool also to detect quantum phase transitions.

\begin{figure}[ht]
\centering {\includegraphics[clip,scale=0.3]{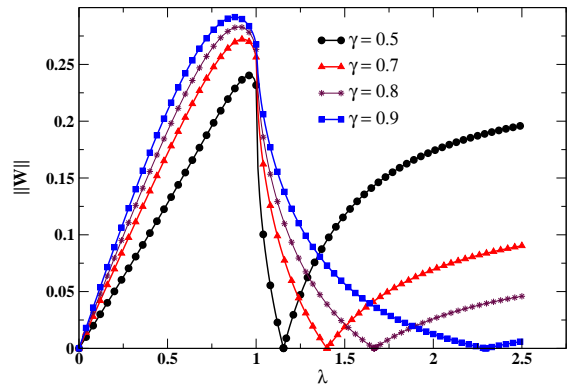}}
\caption{(Color online) Witness $\| {\cal W} \|$ for a pair of spins in the ground state with SSB of the XY 
spin-1/2 infinite chain as a function of $\lambda$ for several values of the parameter $\gamma$. }
\label{f2}
\end{figure}

\section{Quantum correlated states in the Ashkin-Teller model} 
The Ashkin-Teller model has been introduced as a generalization of the Ising spin-1/2 
model to investigate the statitiscs of lattices with four-state 
interacting sites~\cite{Ashkin:43}. It exhibits a rich phase diagram~\cite{ATpd} 
and has recently attracted a great deal of attention  due to several applications, 
e.g., nonabelian anyon models~\cite{Gils:09} and orbital current loops in CuO2-plaquettes 
of high-Tc cuprates~\cite{Gronsleth:09}. The Hamiltonian for the quantum Ashkin-Teller model in 
one-dimension for a chain with $M$ sites is given by 
\begin{eqnarray}
&&H_{AT} =-J\sum_{j=1}^{M}\left( \sigma_{j}^{x}+\tau_{j}^{x} + \Delta
\sigma_{j}^{x} \tau_{j}^{x}\right)  \nonumber \\
&&\hspace{-0.5cm}-J\,\beta \sum_{j=1}^{M}\left( \sigma_{j}^{z}\sigma_{j+1}^{z}
+ \tau_{j}^{z}\tau_{j+1}^{z} + \Delta \sigma_{j}^{z}\sigma_{j+1}^{z}
\tau_{j}^{z}\tau_{j+1}^{z}\right),  \label{at}
\end{eqnarray}
where $\sigma_j^\alpha$ and $\tau_j^\alpha$ $(\alpha = x,y,z)$ are
independent Pauli spin-1/2 operators, 
$J$ is the exchange coupling constant, 
$\Delta$ and $\beta$ are (dimensionless) parameters, 
and periodic boundary conditions (PBC) are adopted, i.e., $%
\sigma^\alpha_{M+1} = \sigma^\alpha_{1}$ and $\tau^\alpha_{M+1} =
\tau^\alpha_{1}$ ($\alpha = x,y,z$). The Ashkin-Teller model is $Z_2 \otimes
Z_2$-symmetric, with the Hamiltonian commuting with the parity operators 
\begin{equation}
\mathcal{P}_1 = \otimes_{j=1}^{M} \sigma_j^x \hspace{1cm} {\text{and}} \hspace{%
1cm} \mathcal{P}_2 = \otimes_{j=1}^{M} \tau_j^x .  \label{parity-at}
\end{equation}
Therefore, the eigenspace of $H_{AT}$ can be decomposed into four disjoint
sectors labelled by the eigenvalues of $\mathcal{P}_1$ and $\mathcal{P}_2$. 
We then numerically diagonalize $H_{AT}$ (via power method) to evaluate the witness $\| {\cal W} \|$, which provides 
a sufficient condition for the absence of $Z_2 \otimes Z_2$-symmetric classical states for finite values of $\beta$ and $\Delta$. 
We illustrate this result by considering different configurations of multipartite states. More specifically, we choose a spin quartet, 
namely, a block of $L=4$ particles described by spin operators $\bigotimes_{k=j}^{j+1} \hat{\sigma}_k \hat{\tau}_k$, and 
respective generalization for a spin octet. For such cases, the absence of classical states for any finite $\Delta$ is displayed 
in Fig.~\ref{f3} for $\beta=1$ in a chain with $N=16$ particles. This is confirmed in the inset of Fig.~\ref{f3}, where the exact 
behavior of global QD is exhibited.
\vspace{-0.2cm}
\begin{figure}[ht]
\centering {\includegraphics[scale=0.32]{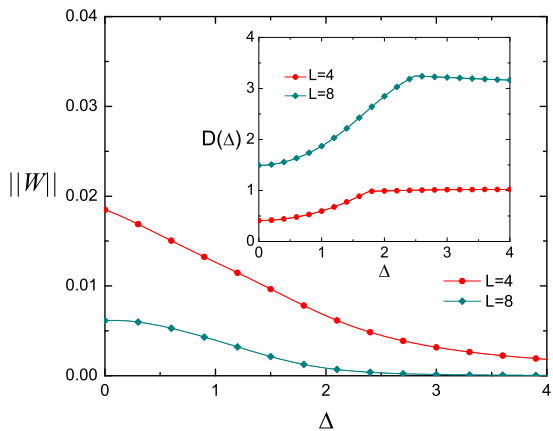}}
\caption{(Color online) Witness $\| {\cal W} \|$ as a function of $\Delta$ for several $L$-spin 
multipartite states in the Ashkin-Teller model for $\beta=1$ in a chain with $N=16$ spin particles.}
\label{f3}
\end{figure}

Moreover, besides sufficent for identifying no classical states in the thermal ($Z_2 \otimes Z_2$-symmetric) ground state 
of the Ashkin-Teller chain, the witness $\| {\cal W} \|$ can also be applied to identify second-order 
quantum phase transitions driven by the parameter $\beta$. We illustrate this result by plotting $\| {\cal W} \|$ for a quartet
in Fig.~\ref{f4} as a function of $\beta$ for $\Delta=3$. From this plot, a second-order critical point can be readout 
through a minimum in the derivative of $\| {\cal W} \|$ (see inset in Fig.~\ref{f4}), which gets pronounced as the length 
of the chain is increased (we observe that the maximum appearing in the plot for $d \| {\cal W} \| /d\beta$ does {\it not} 
indicate any transition, since it does not exhibit pronunciation as we increase the size of the chain). This pronounced minimum, 
which is a typical precursor of the nonanalyticity at the thermodynamic limit, occurs for $\beta=0.61$, which is compatible 
with the quantum phase diagram in Ref.~\cite{ATpd}. A further transition would be expected to occur near $\beta=1.5$. 
This critical point is not apparent from $\| {\cal W} \|$, at least for chains up to $N=20$ spins.
Indeed, this second critical point is usually harder to characterize (see, e.g., characterization 
with entanglement in Ref.~\cite{Rulli:10} and with global QD in the upper insets of Fig.~\ref{f4}, 
where pronunciation is subtle). 
\begin{figure}[ht]
\centering {\includegraphics[scale=0.32]{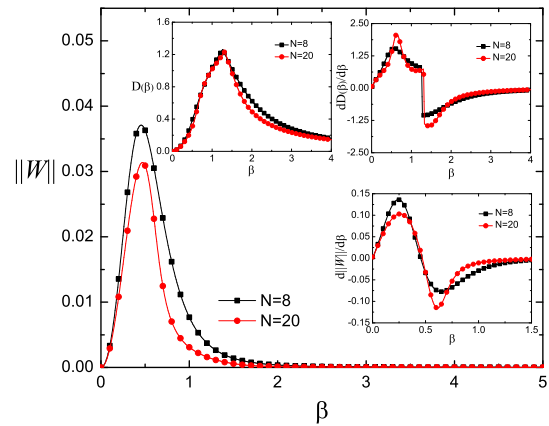}}
\caption{(Color online) Witness $\| {\cal W} \|$ as a function of $\beta$ for a quartet 
in the Ashkin-Teller model for $\Delta=3$ in chains with $N=8$ and $N=20$ spin particles.}
\label{f4}
\end{figure}
\section{Conclusion} 

In summary, we have proposed a multipartite witness for nonclassical states  
based on their disturbance under local measurements. This witness provides a 
sufficient condition for nonclassicality that is in agreement with a nonvanishing 
global QD. As an illustration, we derived the necessary conditions for classicality 
in terms of correlation functions for $Z_2$-symmetric density operators and apply 
the witness for both the XY and Ashkin-Teller chains. Remarkably, the witness detects the 
factorized ground state of the XY chain as its unique classical state (for finite 
value of $\lambda$ and $\gamma$) as the SSB was taken into account. Then, in this example, 
$\| {\cal W} \|$ turns out to be not only sufficient but also necessary for nonclassicality. 
Moreover, the witness was also shown to be able to identify second-order quantum 
phase transitions through a nonanalyticity in the first derivative of $\| {\cal W} \|$ for both XY and 
Ashkin-Teller chains. Even though classicality is rare in many-body microscopic systems~\cite{Ferraro:10}, 
it is rather useful, as a resource to both quantum information and condensed matter applications, to characterize 
the existence of nonclassical correlation through simple and direct procedures. In particular, the typicality 
of states with vanishing genuine $n$-partite correlations would be a possible application of our results. 
The witness $\| {\cal W} \|$ may constitute a helpful tool for such kind of characterization.  
 

\subsection*{Acknowledgments}

We  acknowledge financial support from the Brazilian agencies CNPq and FAPERJ. 
This work was performed as part of the Brazilian National Institute for Science and 
Technology of Quantum Information (INCT-IQ).

\end{document}